\documentclass[12pt]{article}
\usepackage[T1]{fontenc}
\usepackage[
left=0.6in,
right=0.6in,
top=0.6in,
bottom=0.6in
]{geometry}   
\usepackage{setspace}
\usepackage{amsmath,amssymb} 
\usepackage{amsthm}

\usepackage{algorithm} 
\usepackage{algpseudocode}
\usepackage{float}
\usepackage{placeins}

\usepackage{amsmath}
\usepackage{amsfonts}
\usepackage{amssymb}
\usepackage{amsthm}
\usepackage{nicefrac}
\usepackage{mathrsfs}

\usepackage{graphicx}
\graphicspath{{media/}}

\usepackage{booktabs}
\usepackage{array}
\usepackage{tabularx}
\usepackage{makecell}
\usepackage{adjustbox}
\usepackage{subcaption}

\usepackage{algorithm}
\usepackage{algpseudocode}

\usepackage{xcolor}
\usepackage{soul}
\usepackage{enumerate}
\usepackage{changepage}
\usepackage{microtype}

\usepackage{url}
\usepackage{hyperref}

\usepackage{parskip}
\usepackage{mathrsfs}
\usepackage{graphicx}
\usepackage{xcolor}
\usepackage{subcaption}
\usepackage{forest}
\usepackage{amsmath}
\usepackage{amssymb}
\usepackage{hyperref}
 \usepackage{booktabs}
\usepackage{longtable}
\usepackage{array}
\usepackage{tabularx}

\newtheorem{lemma}{Lemma}
\newtheorem{theorem}{Theorem}

\title{
Budget-Constrained Fault-Tolerant Mutual Visibility
for Autonomous Robots under the Mobility Fault Model
}

\author{
Prakhar Shukla,
Animesh Maiti\thanks{Corresponding author: p22ma201@iitj.ac.in},
Shivam Kumar,
Subhash Bhagat
\\[1ex]
Department of Mathematics\\
Indian Institute of Technology Jodhpur\\
Jodhpur, Rajasthan, India\\[1ex]
}

 \date{}

\begin{document}

\maketitle

\begin{abstract}
We investigate the mutual visibility problem for a swarm of $n\ge3$ autonomous mobile robots under the budget-constrained mobility fault model. The robots are opaque, so if three robots are collinear, the middle robot obstructs the visibility between the other two. Each robot is assigned a finite movement budget, reflecting its limited energy, that bounds the total distance it may traverse during the execution. Moreover, an arbitrary number of robots may become permanently immobile due to mobility faults. The objective is to design a distributed algorithm that enables the non-faulty robots to coordinate their movements so that, within a finite time, every non-faulty robot attains unobstructed visibility of all robots in the system, including the faulty ones, while respecting the prescribed movement budget. We consider luminous robots operating under the $\mathsf{SSYNC}$ model with non-rigid movements, without any agreement on their local coordinate systems, and equipped only with a {\it common fixed reference point}. We present a deterministic distributed algorithm that solves the problem despite an arbitrary number of mobility faults. The algorithm guarantees mutual visibility for the non-faulty robots, respects the movement budget of every robot, provides collision-free movements for the robots, and uses only 12 light colors.
\end{abstract}

\section{Introduction}
We consider a finite set $\mathcal{R}$ of $n$ anonymous, homogeneous, opaque, and luminous point robots deployed in the Euclidean plane. Robots have unlimited visibility, but their line of sight is obstructed by opaque robots. Each robot has its own local coordinate system, with no agreement on the coordinate axes or chirality, and executes the same deterministic algorithm. Communication is achieved only by observing the positions and light colors of visible robots. Robots are oblivious except for persistent lights operating under the \textsc {FALL} model. Each robot is equipped with a persistent light that can assume a color from a finite set shared by all robots. The selected color remains unchanged across subsequent activations until the robot explicitly changes it. Thus, the light serves both as persistent local memory and as a means of communication among the robots. The robots operate under the semi-synchronous ($\mathsf{SSYNC}$) model. Under the $\mathsf{SSYNC}$ model, time is divided into discrete rounds. In each round, the scheduler activates an arbitrary nonempty subset of robots, while the remaining robots stay stationary. All robots activated in the same round simultaneously execute one complete \textsc{Look--Compute--Move} cycle. During the \textsc{Look} phase, each
activated robot takes a snapshot containing the positions and light colors
of all robots visible to it, together with its own current light color.
During the \textsc{Compute} phase, it uses the obtained snapshot to compute
a destination point according to the algorithm and may also update its
light color. Finally, during the \textsc{Move} phase, the robot moves toward
its computed destination. The scheduler is fair, meaning that every robot
is activated infinitely often. An epoch is the shortest sequence of
consecutive rounds during which every robot is activated at least once.
A robot may be activated multiple times within the same epoch. By the
fairness of the scheduler, every epoch is finite. Robot movements are non-rigid, i.e., a robot may be stopped before reaching its computed destination after moving at least an unknown positive distance $\delta$. 
In addition, every robot can detect a common fixed external reference
beacon $P^*$ and obtain its relative position in its local coordinate
system during each \textsc{Look} phase.
 
Each robot is initially assigned the same finite movement budget $B$, which bounds the total Euclidean distance it may traverse during the execution. A robot may suffer a \emph{mobility fault} at any time, after which it becomes permanently immobile while retaining its sensing, computation, and light capabilities. Faulty robots cannot distinguish themselves from fault-free robots. Throughout the paper, we consider the budget-constrained mobility fault model, in which an arbitrary number of robots may become mobility faulty.


\subsection{Related Work}

The mutual visibility problem for autonomous mobile robots was formally introduced by Di Luna \emph{et al.}~\cite{di2014mutual} in the classical oblivious robot model. Under obstructed visibility, robots are considered opaque, and hence a robot lying on the line segment joining two other robots prevents them from seeing each other. The objective is therefore to transform an arbitrary initial configuration into one in which no three robots are collinear, while avoiding collisions. Di Luna \emph{et al.} studied the problem for oblivious robots under the $\mathsf{SSYNC}$ scheduler and proposed a convex-hull-based solution. Closely related to mutual visibility is the problem of forming a \emph{general position}. Bhagat \emph{et al.}~\cite{bhagat2016general} studied this problem for asynchronous opaque robots under one-axis agreement and provided a collision-free algorithm that removes collinearities. Early investigations also focused on the efficiency of mutual-visibility algorithms. Sharma \emph{et al.}~\cite{sharma2015bounds} analyzed bounds on existing mutual-visibility algorithms, while subsequent work improved the running time for oblivious robots; in particular, Sharma \emph{et al.}~\cite{sharma2019oblivious} obtained an $\mathcal{O}(n)$-time algorithm in the fully synchronous setting.

A major line of research considers the \emph{luminous robot model}, in which each robot carries a persistent externally visible light whose color can be used both as a small amount of persistent memory and as a means of communication. Di Luna \emph{et al.}~\cite{di2017mutual} gave fundamental results for this model, showing that mutual visibility can be solved using two colors in the rigid $\mathsf{SSYNC}$ setting and three colors in $\mathsf{ASYNC}$; they also considered non-rigid movements and established solutions under additional assumptions such as one-axis agreement. Sharma \emph{et al.}~\cite{sharma2015mutual} subsequently studied the number of colors required and obtained color-efficient solutions. In parallel, Vaidyanathan \emph{et al.}~\cite{vaidyanathan2015logarithmic} investigated the time complexity of complete visibility and proposed a logarithmic-time algorithm for robots with lights. Sharma \emph{et al.}~\cite{sharma2016complete} further reduced the running time to $\mathcal{O}(1)$ rounds in synchronous settings. Bhagat and Mukhopadhyaya~\cite{bhagat2017optimum} considered an energy-oriented variant and proposed an algorithm that restricts the number of movements performed by each robot. For the asynchronous setting, constant-time complete visibility was subsequently obtained by Sharma \emph{et al.}~\cite{sharma2021constant}. Thus, much of the luminous-robot literature has concentrated on reducing either the number of colors, the number of movements, or the overall running time.

Mutual visibility has also been investigated under alternative computational models and geometric domains. Bhagat and Mukhopadhyaya~\cite{bhagat2019persistent} considered the $\mathsf{FSTATE}$ model, in which robots have a constant amount of persistent internal memory rather than externally visible lights, and showed that one bit of memory is sufficient in the semi-synchronous setting. This direction was further developed by Bhagat~\cite{bhagat2020optimum}, who studied optimum mutual visibility with finite-state robots. The problem has also been extended from point robots to robots having physical extent. Sharma \emph{et al.}~\cite{sharma2018fatlights} studied complete visibility for fat luminous robots, whereas Sharma \emph{et al.}~\cite{sharma2018fat} investigated the corresponding problem for non-transparent fat robots without lights. More recently, Bose \emph{et al.}~\cite{bose2023fatcamera} considered a more realistic sensing model in which a fat robot is equipped with a slim omnidirectional camera and proposed a mutual-visibility algorithm under the $\mathsf{SSYNC}$ scheduler. Another line of work restricts robot movements to discrete environments. Adhikary \emph{et al.}~\cite{adhikary2022grid} studied asynchronous luminous robots on an infinite grid, where movements are restricted to grid edges, and Kim \emph{et al.}~\cite{kim2025grid} subsequently investigated two-color complete-visibility algorithms on the grid under stronger coordinate-agreement assumptions. These results demonstrate that the difficulty of mutual visibility depends strongly on the robots' geometric representation, movement domain, memory, and degree of coordinate agreement.

Fault-tolerant mutual visibility constitutes another important branch of the literature. Aljohani and Sharma~\cite{AljohaniS18} considered crash faults and gave a complete-visibility algorithm tolerating a faulty robot under coordinate-agreement assumptions. Poudel \emph{et al.}~\cite{poudel2021fault} studied the more general mobility-fault model, in which a faulty robot loses its ability to move but can continue to execute the other phases of its Look--Compute--Move cycle. Their algorithm works under one-axis agreement and tolerates multiple mobility faults. A different fault model was considered by Sharma~\cite{Sharma18}, where faults affect the lights of the robots rather than their mobility. More recently, Pramanick \emph{et al.}~\cite{pramanick2024angular} investigated mutual visibility in the presence of angular inaccuracies in robot movements and also considered the simultaneous presence of mobility faults. Their algorithms work for luminous robots under one-axis agreement and show that mutual visibility can still be obtained despite bounded directional errors.

The assumption of a common axis was subsequently relaxed in fault-tolerant settings. Pramanick \emph{et al.}~\cite{pramanick2025fault} studied mobility failures for completely disoriented robots, i.e., robots having no agreement on their local coordinate systems. They established fault-tolerant algorithms for multiple faults in the $\mathsf{FSYNC}$ setting and for a single fault in the $\mathsf{SSYNC}$ and $\mathsf{ASYNC}$ settings. More recently, Pramanick \emph{et al.}~\cite{pramanick2026asynchronous} further investigated asynchronous fault-tolerant mutual visibility for disoriented robots. Along a different direction, Bhagat \emph{et al.}~\cite{bhagat2025min} studied the \emph{min-move mutual visibility} problem under the strong combination of disorientation, $\mathsf{ASYNC}$ scheduling, and non-rigid movements. They showed that the problem is impossible without additional persistent information and proposed a luminous algorithm in which every robot moves at most once.

From the viewpoint of solution techniques, a large fraction of the above algorithms are \emph{convex-hull-based}: robots progressively expand, contract, or otherwise manipulate the convex hull until all relevant robots occupy suitable hull vertices~\cite{di2014mutual,di2017mutual,vaidyanathan2015logarithmic,sharma2016complete,sharma2015mutual,sharma2021constant,sharma2018fatlights,sharma2018fat,pramanick2025fault}. Although this approach is powerful, it generally requires robots to travel to geometrically prescribed hull positions and therefore provides limited control over the amount of movement performed by an individual robot. Non-convex-hull techniques instead remove collinearities by carefully selecting local destinations without requiring all robots to become vertices of a common convex hull. Such approaches have been used for general-position formation and movement-efficient mutual visibility~\cite{bhagat2016general,bhagat2017optimum,bhagat2025min}. However, existing non-convex-hull solutions either rely on additional coordinate agreement, do not tolerate mobility failures, or address substantially different fault assumptions.

To the best of our knowledge, no existing algorithm simultaneously provides a \emph{non-convex-hull-based} solution for disoriented robots under mobility faults while also enforcing an explicit budget on robot movement. In this paper, we address this gap by studying the mutual visibility problem under a \emph{budget-constrained mobility fault model}. Our algorithm follows a non-convex-hull strategy and is designed so that the motion of every non-faulty robot remains compatible with the prescribed movement budget. This additional constraint fundamentally distinguishes our setting from earlier mobility-fault algorithms~\cite{poudel2021fault,pramanick2025fault,pramanick2026asynchronous}: directly moving robots toward convex-hull vertices may require an arbitrarily large displacement relative to the available budget, and hence the existing convex-hull-based techniques cannot be applied without substantial modification.

\subsection{Motivation and Our Contribution}

We study fault-tolerant mutual visibility under a finite movement budget, where
each robot can traverse a total Euclidean distance of at most $B$. Existing
fault-tolerant solutions either rely on one-axis agreement or use convex-hull
based relocation strategies. In contrast, we assume no agreement on local
coordinate axes or chirality. Instead, every robot can identify a common fixed
reference point $P^*$ lying outside the convex hull of the initial configuration.
The point $P^*$ provides only a common geometric landmark; it does not induce a
common direction, orientation, or coordinate system.

We consider a strong fault-tolerant visibility requirement: every non-faulty
robot must eventually have unobstructed visibility of every robot in the
system, including mobility-faulty robots. We present a deterministic
non-convex-hull-based algorithm for luminous robots operating under $\mathsf{SSYNC}$ with
non-rigid movements. The algorithm uses the relative distances of robots from
$P^*$ to coordinate movements and resolve cyclic dependencies. Each movement
is chosen using local geometric clearances and the robot's remaining movement
budget, so no robot is required to reach a prescribed distant target.

Our algorithm tolerates an arbitrary number of permanent mobility faults,
guarantees collision-free movements, respects the movement budget of every
robot, uses $12$ light colors, and terminates in $O(n)$ epochs. Table~\ref{tab:survey} summarizes the main differences between the existing mutual visibility algorithms under mobility faults and our proposed approach.

\begin{table*}[ht]
\caption{Existing work on mutual visibility algorithms under mobility fault.}
\label{tab:survey}
\centering
\small
\setlength{\tabcolsep}{4pt}
\renewcommand{\arraystretch}{1.15}

\begin{tabularx}{\textwidth}{@{}
>{\raggedright\arraybackslash}p{2.8cm}
>{\centering\arraybackslash}p{0.9cm}
>{\centering\arraybackslash}p{1.3cm}
>{\raggedright\arraybackslash}X
>{\centering\arraybackslash}p{0.9cm}
>{\centering\arraybackslash}p{1.2cm}
>{\raggedright\arraybackslash}p{2.7cm}
@{}}
\toprule

Reference
& Faults
& Scheduler
& Axis/assumption
& Colors
& Time
& Approach
\\

\midrule

Poudel et al.~\cite{poudel2021fault}
& \(f\)
& \(\mathsf{ASYNC}\)
& One-axis agreement
& \(3\)
& \(O(n)\)
& Non-convex hull
\\

Pramanick et al.~\cite{pramanick2026asynchronous}
& \(f\)
& \(\mathsf{ASYNC}\)
& No axis agreement
& \(35\)
& \(O(n)\)
& Convex hull
\\

\midrule

\textbf{This work}
& \(\boldsymbol{f}\)
& \(\boldsymbol{\mathsf{SSYNC}}\)
& \textbf{Beacon and uniform movement budget}
& \(\mathbf{12}\)
& \(\boldsymbol{O(n)}\)
& \textbf{Non-convex hull}
\\

\bottomrule
\end{tabularx}

\end{table*}

 \section{Model and Definitions}

We consider the following robot model and terminologies to design our algorithms:

\begin{itemize}
    \item \textbf{Robots:} Let $\mathcal{R}=\{r_1,r_2,\ldots,r_n\}$ be a set of $n$ anonymous, homogeneous, autonomous, opaque, and luminous point robots deployed in the Euclidean plane. The robots do not know the value of $n$, have unlimited visibility range, and communicate only through observations. Due to opacity, if three robots are collinear, then the middle robot obstructs the visibility between the other two.
    \item For two points  $x$ and $y$ in the plane, let $\|x-y\|$
denote the Euclidean distance between $x$ and $y$.
    \item \textbf{Computation Model:} Robots execute the \textsc{Look--Compute--Move} cycle under the semi-synchronous ($\mathsf{SSYNC}$) model with a fair scheduler. Robot movements are non-rigid. There exists a fixed but unknown
constant $\delta>0$ such that, if a non-faulty robot $r_i$ computes
a destination $d_i$ and $\|r_i-d_i\|\leq\delta,$
then it reaches $d_i$; otherwise, the scheduler may stop it after
any displacement in $[\delta,\|r_i-d_i\|]$.
Thus, whenever a non-faulty robot computes a destination distinct
from its current position, it performs a strictly positive movement. Robots are oblivious except for persistent externally visible lights operating under the \textsc{FALL} model.

\item \textbf{Reference Beacon:}
The robots have access to a common fixed external reference beacon
$P^*$. The beacon is not a robot and is detected through an external
localization or beacon-sensing mechanism; hence its detection is not
affected by the opacity of the robots. During every Look operation,
robot $r_i$ obtains the relative position of $P^*$ expressed in its
current local coordinate system. Thus, $P^*$ provides a common
geometric landmark, but does not induce a common coordinate system,
axis orientation, handedness, or chirality. We assume that $P^*$ is
distinct from every initial robot position and lies strictly outside the convex hull of the initial robot configuration. We henceforth refer to $P^*$ simply as the {\it reference point}.


    \item \textbf{Movement Budget:} Each robot is assigned the same finite
movement budget $B>0$, representing the maximum Euclidean distance
it may traverse during the execution. At every activation, robot
$r_i$ can locally obtain its current remaining movement budget,
denoted by $B_i(t)$. Initially, $B_i(0)=B$. Only the distance actually
traveled is deducted from $B_i(t)$. Thus, interrupted non-rigid
movements consume only the executed portion, while a movement
attempt by a mobility-faulty robot that fails to initiate consumes
no budget. If $\ell_i$ denotes the total distance traveled by $r_i$,
then $\ell_i\leq B$.


    \item \textbf{Mobility Faults:} A robot may develop a permanent mobility fault at any stage of the execution. A mobility-faulty robot can continue to perform the \textsc{Look} and \textsc{Compute} phases and update its light, but it can no longer move. Faulty robots cannot distinguish themselves from fault-free robots. In contrast to a non-rigid interruption, a mobility fault may result
in zero executed displacement even when the computed destination is
different from the robot's current position.

 \item \textbf{Configuration:}
The position of robot $r_i$ at time $t$ is denoted by $r_i(t)$,
and its light color by $c_i(t)$. The \emph{geometric configuration}
at time $t$ is
\[
\mathcal{R}(t)=\{r_1(t),r_2(t),\ldots,r_n(t)\}.
\]
The complete robot configuration is
\[
\mathcal{C}(t)=
\{(r_i(t),c_i(t)):1\leq i\leq n\}.
\]
Thus, $\mathcal{R}(t)$ describes only the robot positions, whereas
$\mathcal{C}(t)$ also records the persistent light states that govern
the execution of the algorithm. A configuration is called \emph{linear} if all robot positions lie
on a single line, and \emph{non-linear} otherwise. Let
$\widetilde{R}_{L}$ and $\widetilde{R}_{NL}$ denote, respectively, the
sets of linear and non-linear geometric configurations.

    \item \textbf{Collinearity:} A line containing at least three robot positions is called a \emph{line of collinearity}. A robot is \emph{non-terminal} on such a line if it lies in the interior of the segment joining two other robot positions on the line; otherwise, it is \emph{terminal}. A non-terminal robot lying on at least two lines of collinearity is called a \emph{junction robot}. For each robot $r_i$, let $\mathcal{L}_i^0$ denote the set of
initial lines of collinearity containing $r_i(0)$.

    \item \textbf{Visibility:} The vision of robot $r_i$ at time $t$ is the set $\mathcal{V}_i(t)$ consisting of all robot positions visible to $r_i$, including its own position.

\begin{figure}[htbp]
\centering
\includegraphics[width=0.5\linewidth]{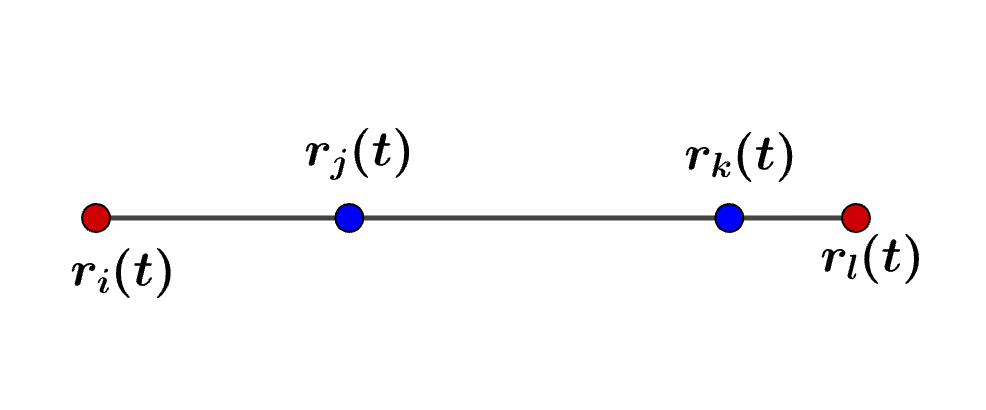}
\vspace*{-0.5cm}
\caption{Illustration of adjacent and friend robots. The adjacent robot pairs are $(r_i,r_j)$, $(r_j,r_k)$, and $(r_k,r_l)$. Since $r_j$ and $r_k$ are non-terminal robots, $r_i$ and $r_k$ are the friend robots of $r_j$, whereas $r_j$ and $r_l$ are the friend robots of $r_k$.}
\label{fig:fri_adj}
\end{figure}

    \item {\bf Adjacent and friend robots:} Two robots $r_i$ and $r_j$ are said to be \emph{adjacent} at time $t$ if they are collinear and no robot lies in the open line segment joining $r_i(t)$ and $r_j(t)$. Equivalently, they are mutually visible along the line containing them. Furthermore, let $r_i$ be a non-terminal robot on a line of collinearity $\mathcal L$ at time $t$, and let $r_j(t)$ and $r_k(t)$ lie on the two opposite sides of $r_i(t)$ along $\mathcal L$, with both robots visible to $r_i$. Then $r_j$ and $r_k$ are called the \emph{friend robots} of $r_i$ with respect to $\mathcal L$. (see Figure~\ref{fig:fri_adj})

    \item \textbf{Cyclic Dependency:} A set of lines of collinearity forms a \emph{cyclic dependency} if consecutive lines intersect at distinct junction robots and the first and last lines also intersect, thereby forming a cycle. The shared junction robots are called \emph{critical non-terminal robots}. If the first and last lines do not intersect, the collection forms a \emph{non-cyclic dependency chain}. (see Figure~\ref{fig:cyclicdep})
\end{itemize}

\textbf{Problem Definition.}
Let $\mathcal R=\{r_1,\ldots,r_n\}$, $n\geq 3$, be a set of robots initially occupying distinct positions in the Euclidean plane. An arbitrary number of robots may suffer permanent mobility faults during the execution. Let $\mathcal R_{NF}(t)$ denote the set of non-faulty robots at time $t$.

The \emph{budget-constrained strong fault-tolerant mutual visibility problem}
requires that, from every admissible initial configuration, there exists a
finite time $T$ such that every non-faulty robot has unobstructed visibility
of every other robot in the system. Formally, for every
$r_i\in \mathcal R_{NF}(T)$ and every $r_j\in \mathcal R\setminus\{r_i\}$,
\[
\operatorname{int}\!\left(\overline{r_i(T)r_j(T)}\right)
\cap
\left(\mathcal R(T)\setminus\{r_i(T),r_j(T)\}\right)
=\emptyset .
\]
Moreover, if $\ell_i$ is the total Euclidean distance traversed by $r_i$, then $\ell_i\leq B
\ \text{for every } r_i\in \mathcal R.$ Thus, unlike the standard fault-tolerant mutual visibility problem, a non-faulty robot must eventually see not only every non-faulty robot but also every mobility-faulty robot.

\begin{figure*}[ht]
    \centering

    \begin{subfigure}[t]{0.42\textwidth}
        \centering
        \includegraphics[width=0.75\linewidth]{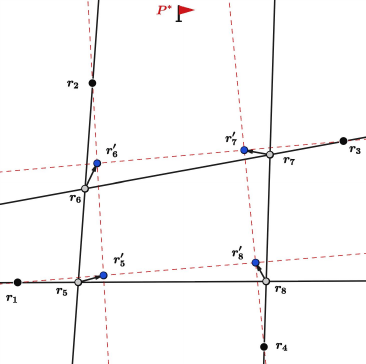}
        \caption{}
        \label{fig:cyclic-four}
    \end{subfigure}
      \hspace{0.03\textwidth}
    \begin{subfigure}[t]{0.42\textwidth}
        \centering
        \includegraphics[width=0.75\linewidth]{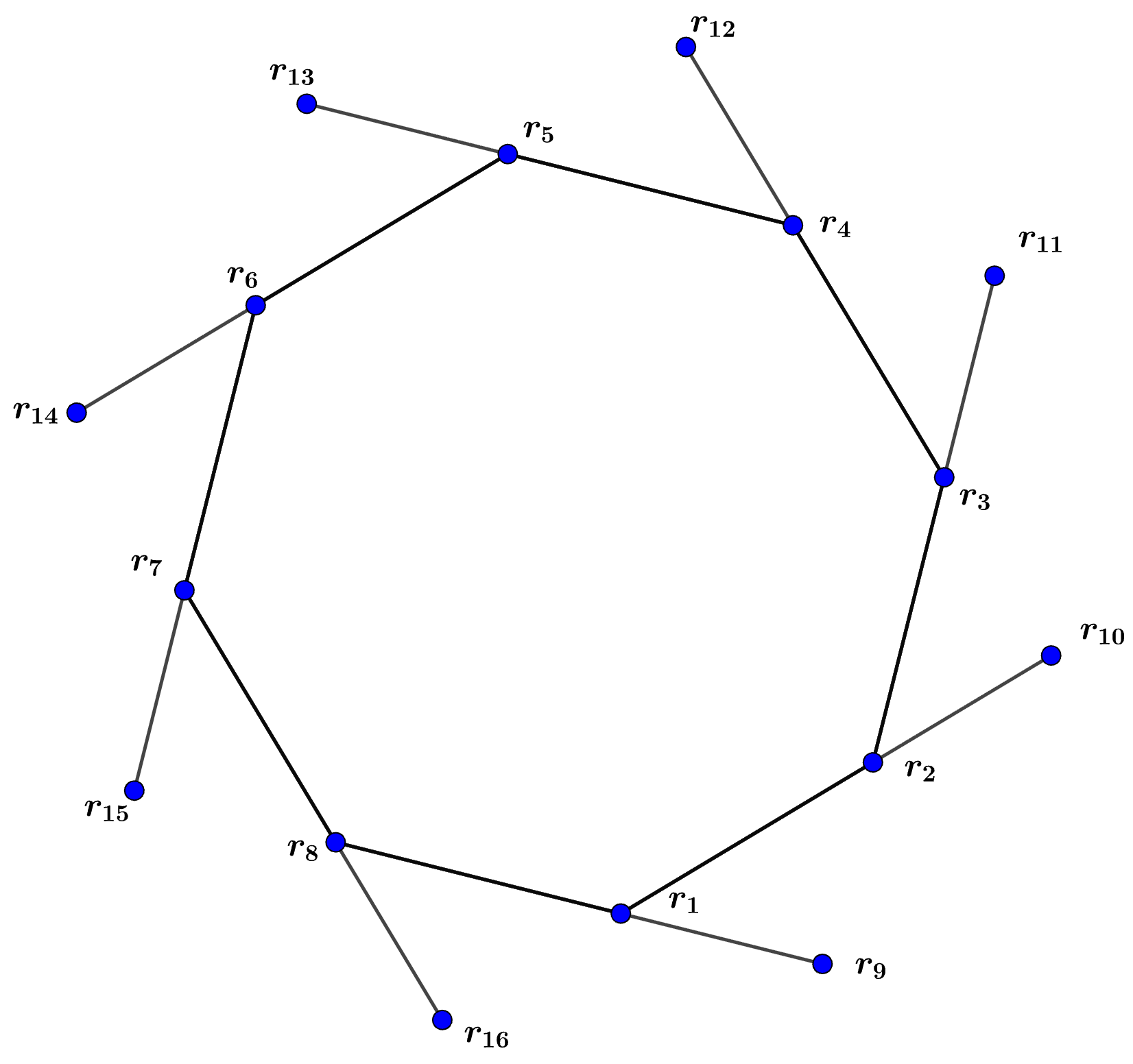}
        \caption{}
        \label{fig:cyclic-critical}
    \end{subfigure}

    \caption{Illustration of cyclic dependencies:
    (a) a cyclic configuration involving four non-terminal robots, where the black robots are faulty, and
    (b) a cyclic dependency involving multiple critical non-terminal
    robot positions.}
    \label{fig:cyclicdep}
\end{figure*}

\section{Algorithm Mutual-Visibility-Mobility-Fault()}

The algorithm consists of four logical stages encoded by the light colors.
These stages need not start or terminate simultaneously for all robots.
In particular, a robot never needs to infer that a phase has been completed
by all robots in the system. Instead, it advances according to its own light
state and the states of the visible robots on which its current movement
depends. Such robots are adjacent or friend robots and hence are visible by
definition. Therefore, all state transitions used by the algorithm are
locally detectable.

\subsection{ Computation of destination point}
\label{comnputation}
Before describing the four phases, we define the destination computed
by an eligible robot. Let $\operatorname{dest}_i(t)$ denote the
destination of robot $r_i$ at time $t$. Its destination is determined
by a movement direction $\operatorname{DIR}_i(t)$ and a displacement
$\sigma_i(t)$. The construction below guarantees that the movement
respects the remaining budget, avoids collisions, and prevents the
creation of collinearities among robots that were initially
non-collinear.
 We first define a few terms needed to compute the destination points of the robots. 
 \begin{itemize}
     \item Let $r_i(t)$ and $r_j(t)$ be two distinct robot positions in $\mathcal R(t)$. We denote the straight line joining these two robot positions by   $\mathcal L_{ij}(t)$. 
     \item Let $r_i(t)$, $r_j(t)$ and $r_k(t)$ be  three non-collinear robot positions in $\mathcal R(t)$.  We denote the {\it perpendicular distance} of the line $\mathcal L_{ij}(t)$ from the point $r_k(t)$ by $d^k_{ij}(t)$. 
     
\item
For every robot position $r_m(t)\in V_i(t)\setminus\{r_i(t)\}$,
let $L_m^*(t)$ denote the line passing through $r_m(t)$ and the
reference point $P^*$. Whenever $r_i(t)\notin L_m^*(t)$, define
\[
d^{P^*}_{i,m}(t)
=
\operatorname{dist}\bigl(r_i(t),L_m^*(t)\bigr).
\]

Suppose that $\mathcal{R}(t)\in\mathcal{R}_{NL}$. Define
\[
\mathcal{D}_i(t)
=
\left\{
d^{P^*}_{i,m}(t):
r_m(t)\in V_i(t)\setminus\{r_i(t)\},
\ r_i(t)\notin L_m^*(t)
\right\}
\]
\[
\qquad\cup
\left\{
d^k_{ij}(t),d^j_{ik}(t),d^i_{jk}(t):
r_j(t),r_k(t)\in V_i(t)\setminus\{r_i(t)\},
\ r_i(t),r_j(t),r_k(t)
\text{ are non-collinear}
\right\}.
\]

We define
\[
d_i(t)=
\begin{cases}
\min \mathcal{D}_i(t),
& \text{if }\mathcal{D}_i(t)\neq\emptyset,\\[1mm]
+\infty,
& \text{if }\mathcal{D}_i(t)=\emptyset.
\end{cases}
\]

Thus, whenever finite, $d_i(t)$ is the minimum perpendicular
clearance induced by the relevant lines considered in the
destination construction.



 
 \item Let $r_i(t)$ be a robot positions. Let $NT_i(t)$ be the set of non-terminal robot positions visible to
$r_i$ such that
\[
\begin{aligned}
NT_i(t)=\Bigl\{\, &
r_k(t)\in V_i(t)\setminus\{r_i(t)\}:\\
&
r_k(t)\text{ is non-terminal and}\\
&
P^*\notin L_{ik}(t)
\,\Bigr\}.
\end{aligned}
\]


For every $r_k(t)\in NT_i(t)$, the line $L_{ik}(t)$ joining
$r_i(t)$ and $r_k(t)$ divides the plane into two open half-planes.
Let $HP^{+}_{ik}(t)$ denote the open half-plane containing $P^*$.
 
 Let $HP^+_{ik}$  be the open half plane that contains $P^*$. Define the safe region of $r_i$ by
\[
\operatorname{Reg}_i(t)
=
\bigcap_{r_k(t)\in NT_i(t)} HP^+_{ik}(t).
\]
If $NT_i(t)=\emptyset$, we set
$\operatorname{Reg}_i(t)=\mathbb{R}^2$.

 For every $r_k(t)\in NT_i(t)$, the segment from $r_i(t)$ toward
$P^*$ locally enters $HP^+_{ik}(t)$; hence a direction toward $P^*$
is contained in the corresponding safe side.
 \end{itemize}


  \medskip
\paragraph{Direction of movement.}
We distinguish the direction used in Phase I from the radial
direction used in Phases II--IV.

For a Phase I movement, 
When $\mathcal{R}(t)\in\mathcal{R}_{NL}$, let $x_i$ be any point in
$\operatorname{Reg}_i(t)$ such that the segment
$\overline{r_i(t)x_i}$ does not lie on any line
$L_{ik}(t)$, where
\[
r_k(t)\in V_i(t)\setminus\{r_i(t)\}.
\] Define
\[
\mathbf{u}^{I}_i(t)
=
\frac{x_i-r_i(t)}
     {\|x_i-r_i(t)\|}.
\]

If $\mathcal{R}(t)\in\mathcal{R}_{L}$ and $\mathcal{L}_0$
denotes the current line containing all robot positions, let
$\mathcal{L}_0^\perp$ be the line perpendicular to
$\mathcal{L}_0$ at $r_i(t)$, and choose an arbitrary point
$y_i\in\mathcal{L}_0^\perp\setminus\{r_i(t)\}$. Define
\[
\mathbf{u}^{I}_i(t)
=
\frac{y_i-r_i(t)}
     {\|y_i-r_i(t)\|}.
\]

For every radial movement in Phases II--IV, define
\[
\rho_i(t)=\|r_i(t)-P^*\|
\]
and
\[
\mathbf{u}^{R}_i(t)
=
\frac{P^*-r_i(t)}
     {\rho_i(t)}.
\]
Thus, $\mathbf{u}^{R}_i(t)$ is the unit vector directed from
$r_i(t)$ towards the reference point $P^*$.
  
  
 \noindent \textbf{The amount of displacement:}  Let $D_i(t)$ denote the minimum distance of any two robot positions in $\{r_i(t), {\mathcal V}_i(t)\}$. Let $B_i(t)$ denote the remaining movement budget locally available to robot $r_i$ at time $t$.
 The amount of displacement of $r_i$ at time $t$ is denoted by $\sigma_i(t)$  and it is defined as follows,
 
For a Phase I movement,
\[
\sigma_i(t)=
\begin{cases}
\dfrac{1}{3^4}
\min\{d_i(t),D_i(t),B_i(t)\},
& \mathcal{R}(t)\in\mathcal{R}_{NL},\\[2mm]
\dfrac{1}{3^4}
\min\{D_i(t),B_i(t)\},
& \mathcal{R}(t)\in\mathcal{R}_{L}.
\end{cases}
\]
For a radial movement toward $P^*$ in Phases II--IV,
\[
\sigma_i(t)=
\begin{cases}
\dfrac{1}{3^4}
\min\{d_i(t),D_i(t),B_i(t),\rho_i(t)\},
& \mathcal{R}(t)\in\mathcal{R}_{NL},\\[2mm]
\dfrac{1}{3^4}
\min\{D_i(t),B_i(t),\rho_i(t)\},
& \mathcal{R}(t)\in\mathcal{R}_{L}.
\end{cases}
\]
The quantity $\sigma_i(t)$ is computed to keep the total movements of each robot within budget and also to avoid the creation of new collinearities with those robots which are not initially collinear.
For every radial movement,
$\sigma_i(t)<\rho_i(t)$; hence the computed destination lies strictly
in the interior of the segment joining $r_i(t)$ and $P^*$, and the
robot neither reaches nor crosses the reference point.
\medskip

\paragraph{Destination point.}
The destination of an eligible robot $r_i$ is defined according
to the phase in which it moves:
\[
\operatorname{dest}_i(t)=
\begin{cases}
r_i(t)+\sigma_i(t)\mathbf{u}^{I}_i(t),
& \text{in Phase I},\\[2mm]
r_i(t)+\sigma_i(t)\mathbf{u}^{R}_i(t),
& \text{in Phases II--IV}.
\end{cases}
\]
For a radial movement,
\[
\sigma_i(t)
\leq
\frac{\rho_i(t)}{3^4}
<
\rho_i(t).
\]
Hence,
\[
\operatorname{dest}_i(t)
\in
\operatorname{int}
\bigl(\overline{r_i(t)P^*}\bigr),
\]
so the robot moves strictly towards $P^*$ without reaching or
crossing the reference point.



\subsection{Algorithm Description}
 Initially, all the robots have the lights {\it off}. During {\it Phase I}, both terminal and non-terminal robots move to break initial collinearities. During {\it Phase II} and {\it Phase III}, only non-terminal robots move to break newly created collinearies, including the cyclic dependencies, and during {\it Phase IV}, only terminal robots move to break newly created collinearities with faulty non-terminal robots.
 
 
\noindent\textbf{Phase I (Both terminal and non-terminal robots move):} Let $r_i$, $r_j$, $r_k$, $r_l$ and $r_m$ be collinear robots on a line $\mathcal L$ such that robots $r_i$ and $r_m$ are terminal, other robots are non-terminal, $r_j$ is adjacent to $r_i$, and $r_l$ is adjacent to $r_m$. As the terminal robots move first, initially only $r_i$ and $r_m$ are eligible to move. Without loss of generality, let us consider the following order of activations of the robots: $r_i$ is activated first, followed by $r_j$, then $r_m$ is activated, followed by $r_l$, and lastly $r_k$ is activated. Upon activation, robot $r_i$ turns its light {\it moving-1} and moves. The {\it moving-1} keeps the information that robot $r_i$ has made its initial move. In the next activation, robot $r_i$ finds itself as the terminal robot with light {\it moving-1} and changes its light to {\it moved-T-1} to let others know that it is currently a terminal robot that has made its initial movement. Now, after the initial movement of $r_i$, robot $r_j$ may become terminal or may remain non-terminal (due to a fault in $r_i$). If $r_j$ becomes terminal, then it acts like $r_i$. Otherwise,  $r_j$ finds one of its friends  with light {\it move-T-1} ($r_i$ and $r_k$ are friends of $r_j$) has light {\it moved-T-1}. We allow robot $r_j$ to move while changing its light to {\it moving-1}. If it becomes terminal, it changes its light to {\it move-T-1}. Otherwise, it changes its light to {\it move-NT-1} to indicate that it is currently a non-terminal robot that has made its initial movement. Now suppose that both $r_j$ and $r_l$ are initially faulty. According to the activation schedule and moving rule, when $r_k$ becomes active, it finds both of its friends with light {\it moved-NT-1}, this let $r_k$ know that all the collinear robots on $\mathcal L$ have moved. We allow the robot $r_k$ to move like a non-terminal robot. Now, a non-terminal robot may find one of its friends as {\it moved-T-1} and another as {\it moved-NT-1}. In this case also, we allow the non-terminal robot with the light {\it off} to move. During Phase I, every eligible robot computes its destination according to Section~\ref{comnputation} and moves toward it. We summarize these movement strategies as follows:

 \medskip
 
 \noindent {\bf $\boldsymbol{r_i}$ is a physically terminal robot:} The robot $r_i$ finds itself as a physically terminal robot and
	\begin{enumerate}
		\item its light is {\it off}: it changes its light to {\it moving-1} and moves.
        \item its light is {\it moving-1}: it changes its light to  {\it moved-T-1} and does not move.
        
	\end{enumerate}
    
\medskip

 \noindent {\bf $\boldsymbol{r_i}$ is a physically non-terminal robot:} Robot $r_i$ finds itself as a physically non-terminal robot and
	\begin{enumerate}
		\item its light is {\it off},  and at least one friend robot in each friend-pair of $r_i$ has light {\it moved-T-1} or {\it moved-NT-1}: it changes its light to {\it moving-1} and moves.
		\item its light is {\it moving-1}: it changes its light to {\it moved-NT-1} and does not move.
	\end{enumerate}

After completing its Phase I action, every robot has light either
\textsc{moved-T-1} or \textsc{moved-NT-1}. A robot proceeds to its
next logical stage according to its own light state and the states of
the visible robots on which its subsequent action depends; no global
Phase I completion test is required.
The movements performed by the robots during Phase~I satisfy several useful properties. For ease of presentation, we first state these properties, without proofs, as observations to facilitate understanding of the algorithm.

\begin{lemma}
\label{l_2}
During the whole execution of Phase I, each non-faulty robot moves exactly once.
\end{lemma}
\begin{proof}
    A robot with light \textsc{off} performs at most one Phase I movement, because after moving it changes its light to \textsc{moving-1} and then to either \textsc{moved-T-1} or \textsc{moved-NT-1}. To show existence, observe that terminal robots are eligible immediately, and their completed movements make the adjacent non-terminal robots eligible. Repeating this argument propagates eligibility from the two ends of every initial line towards its interior. By fairness, every non-faulty robot is eventually activated when eligible and hence performs exactly one Phase I movement.
\end{proof} 

\begin{lemma}
\label{l_3}
During the whole execution of Phase I, at most two robots from the same initial line of collinearity move in a single round.
\end{lemma}
 \begin{proof}
   According to movement strategies of Phase I, on any line of collinearity, eligibility propagates from its two terminal ends towards the interior. Initially, only the two terminal robots can move. After a robot completes its Phase I movement, at most its next non-terminal neighbor on that side can become eligible. Since all robots activated in an $\mathsf{SSYNC}$ round compute from the same pre-movement configuration, eligibility cannot propagate by more than one position from either end in the same round. Hence, at most two robots from a line of collinearity move in any round. 
 \end{proof}

\begin{lemma}
\label{l_4}
Let $r_i,r_j,r_k$ be three robots that are non-collinear in the
initial configuration. Then they remain non-collinear throughout
Phase~I.
\end{lemma}

\begin{proof}
We prove the claim by induction over the rounds of Phase~I.
Suppose that $r_i,r_j,r_k$ are non-collinear immediately before
some round, and let $h>0$ denote the minimum altitude of the
triangle $ \triangle r_i r_j r_k$ at that time. We first show that every robot among
$r_i,r_j,r_k$ that moves in this round is displaced by less than
$h/2$.

Consider robot $r_i$. If both $r_j$ and $r_k$ are visible to
$r_i$, then the three altitudes of
$\triangle r_i r_j r_k$ are among the distances considered in the
computation of $d_i(t)$. Hence, $ d_i(t)\leq h .$ Since the configuration is non-linear, the destination construction
gives
\[
        \sigma_i(t)
        \leq \frac{d_i(t)}{3^4}
        \leq \frac{h}{3^4}
        <\frac{h}{2}.
\]

Now suppose that one or both of $r_j,r_k$ are not visible to $r_i$.
Let $\widehat r_j$ and $\widehat r_k$ be the first visible robots on
the segments $\overline{r_i r_j}$ and $\overline{r_i r_k}$,
respectively; if $r_j$ (respectively, $r_k$) is visible, take
$\widehat r_j=r_j$ (respectively, $\widehat r_k=r_k$).
There exist $\alpha,\beta\in(0,1]$ such that
\[
 \widehat r_j-r_i=\alpha(r_j-r_i),
 \qquad
 \widehat r_k-r_i=\beta(r_k-r_i).
\]

Let $A$ be the area of $\triangle r_i r_j r_k$, and let $M=\max\{
       |r_i r_j|,\,
       |r_i r_k|,\,
       |r_j r_k|
      \}.$ Then $ h=\frac{2A}{M}.$ Let $\widehat h$ be the minimum altitude of
$\triangle r_i\widehat r_j\widehat r_k$, and let
$\widehat M$ be the length of its longest side. Since its area is
$\alpha\beta A$, $\widehat h=\frac{2\alpha\beta A}{\widehat M}.$
Furthermore,
\[
\begin{aligned}
 \alpha\beta M
 &\leq
 \alpha\beta\bigl(|r_i r_j|+|r_i r_k|\bigr)\\
 &\leq
 \alpha |r_i r_j|+\beta |r_i r_k|
 \leq 2\widehat M .
\end{aligned}
\]
Therefore, $ \widehat h\leq 2h.$ Since $\widehat r_j$ and $\widehat r_k$ are visible to $r_i$,
the altitudes of
$\triangle r_i\widehat r_j\widehat r_k$ are considered in the
computation of $d_i(t)$. Thus, $ d_i(t)\leq\widehat h\leq2h,$ and consequently
\[
        \sigma_i(t)
        \leq\frac{d_i(t)}{3^4}
        \leq\frac{2h}{3^4}
        <\frac{h}{2}.
\]
The same argument applies to $r_j$ and $r_k$. A robot that is not
activated or is mobility-faulty has displacement zero, and a
non-rigid interruption can only decrease the displacement.
Therefore, irrespective of which subset of the three robots is
activated in the round, each of them is displaced by less than
$h/2$.

Suppose, for contradiction, that after these simultaneous
movements the three robots become collinear on a line $\ell$.
Since the new position of each robot lies on $\ell$ and its
displacement is less than $h/2$, each of the three positions before
the round is at perpendicular distance less than $h/2$ from
$\ell$. Hence the entire triangle
$\triangle r_i r_j r_k$ before the round is contained in an open
strip of width strictly smaller than $h$ whose axis is $\ell$.
This is impossible: the minimum width of a triangle is equal to its
minimum altitude, which is $h$. Thus $r_i,r_j,r_k$ remain
non-collinear after the round.
Since they are non-collinear initially, induction over all Phase I
rounds proves that they remain non-collinear throughout Phase~I.
\end{proof}

\begin{lemma}
\label{l_5}
If an initial line of collinearity persists at the end of
Phase~I, then every robot lying on that line is mobility-faulty.
\end{lemma}

\begin{proof}
  By Lemma~\ref{l_2}, every non-faulty robot initially on the line performs a positive movement in a direction different from the line and hence
leaves it. Moreover, by Lemma~\ref{l_3}, no robot initially outside the line
can move onto it and become collinear with two robots initially on the
line. Thus, any robot lying on an initial line at the end of Phase~I
must be mobility-faulty. 
\end{proof}

\begin{lemma}
\label{l_6}
During Phase~I, every newly created line of collinearity contains
exactly three robot positions, all of which were collinear in the
initial configuration. Moreover, at least one of these robots is
terminal.
\end{lemma}

\begin{proof}
Let $\mathcal L'$ be a line of collinearity created during Phase~I. By Lemma~\ref{l_4}, robots that were not collinear in the initial configuration cannot become collinear during Phase~I. Hence, all robots on $\mathcal L'$ originate from the same initial line of collinearity $\mathcal L$.
By Lemma~\ref{l_3}, at most two robots of $\mathcal L$ can move in the
same round. Moreover, the destination construction prevents a
moving robot from reaching a line already determined by two
stationary robots. Therefore, $\mathcal L'$ cannot contain more
than three robot positions; since it is a line of collinearity, it
contains exactly three.

Finally, Phase I eligibility propagates from the terminal ends of
$\mathcal L$ towards its interior. Thus, any new collinearity is
created by the current frontier movement(s) together with a robot
already exposed on one of these fronts. By the Phase I movement
order, at least one of the three robots is terminal. Moreover, if all robots of the newly created collinearity are
non-faulty, then at least one of them belongs to the unprocessed
part of the corresponding propagation front and has not yet
performed its Phase I movement.

\end{proof}

\begin{lemma}
\label{l_7}
Every newly created line of collinearity that persists at the end of
Phase I contains at least one mobility-faulty robot.
\end{lemma}

\begin{proof}
Let $\mathcal L'$ be a newly created line of collinearity that persists at the
end of Phase I. By Lemma~\ref{l_6}, $\mathcal L'$ contains exactly three robots
originating from the same initial line of collinearity $\mathcal L$, and at
least one of them is terminal.

Suppose, for contradiction, that all three robots on $\mathcal L'$ are
non-faulty. Phase I eligibility on $\mathcal L$ propagates from its terminal
ends towards the interior. Hence, when an all-non-faulty new
collinearity $\mathcal L'$ is created by the current frontier movement(s), at
least one robot $r$ on $\mathcal L'$ belongs to the unprocessed part of the
corresponding frontier and has not yet performed its Phase I
movement. The robots preceding $r$ on that frontier eventually
display a Phase I completion color. Consequently, $r$ either becomes
terminal or satisfies the friend-pair eligibility condition.

Since $r$ is non-faulty and the scheduler is fair, it is eventually
activated while eligible and performs its positive Phase I movement.
By the destination construction, its movement direction does not
coincide with its current line of collinearity; hence $r$ leaves
$\mathcal L'$. Therefore, $\mathcal L'$ cannot persist until the end of Phase I, a
contradiction.

Thus, if a newly created line of collinearity persists at the end of
Phase I, at least one of its robots must be mobility-faulty.
\end{proof}




\begin{lemma}
\label{l_8}
At the end of Phase~I, all non-faulty non-terminal robots are
mutually visible.
\end{lemma}
\begin{proof}
Suppose, for contradiction, that two non-faulty non-terminal robots
$r_i$ and $r_j$ are not mutually visible at the end of Phase~I.
Then some robot $r_k$ lies in the interior of
$\overline{r_i r_j}$, and hence $r_i,r_k,r_j$ lie on a line of
collinearity $\mathcal L$.

The line $\mathcal L$ cannot be a persisting initial line of
collinearity, since by Lemma~\ref{l_6} every robot lying on such a line
at the end of Phase~I is mobility-faulty, whereas $r_i$ and $r_j$ are
non-faulty. Hence, $\mathcal L$ must be a newly created line of
collinearity.

By Lemma~\ref{l_5}, $\mathcal L$ contains exactly three robots and at
least one of them is terminal. However, $r_i$ and $r_j$ are
non-terminal by assumption, and $r_k$ is also non-terminal since it
lies in the interior of $\overline{r_i r_j}$. Thus all three robots
on $\mathcal L$ are non-terminal, contradicting Lemma~\ref{l_6}.
Therefore, every pair of non-faulty non-terminal robots is mutually
visible at the end of Phase~I.
\end{proof}

\begin{lemma}
\label{l_9}
If two non-faulty terminal robots are not mutually visible at the end of Phase I, then there is a faulty non-terminal robot that lies in between them.
\end{lemma}
\begin{proof}
Let $r_i$ and $r_j$ be two non-faulty terminal robots that are not
mutually visible at the end of Phase~I, and let $r_k$ lie between
them. By Lemma~\ref{l_5}, their line cannot be a persisting initial
line of collinearity; hence it is newly created. By Lemma~\ref{l_6},
this line contains exactly three robots, and by Lemma~\ref{l_7}, at
least one of them is mobility-faulty. Since $r_i$ and $r_j$ are
non-faulty, $r_k$ must be mobility-faulty. Moreover, as $r_k$ lies
between $r_i$ and $r_j$, it is non-terminal.
\end{proof}


\textbf{Phase II (Only non-faulty non-terminal robots move):} During this phase, every non-faulty non-terminal robot moves to break the newly created collinearities during Phase I. Let us consider a non-faulty non-terminal robot $r_i$ on a line of collinearity $\mathcal L$.  Then, by Lemmas~\ref{l_6}--~\ref{l_7} above, $\mathcal L$ is a new line of collinearity created in Phase I, and at least one of its friends on $\mathcal L$ is a faulty robot. If the other friend is a terminal robot or a faulty non-terminal robot, then if it moves, it breaks the collinearity on $\mathcal L$, since terminal robots do not move in this phase. Now, if the other friend is a non-faulty non-terminal robot, then this robot and $r_i$ move simultaneously in the same round, and $r_i$ may remain collinear after the movements. If these robots are part of any cyclic dependency, then this process may continue. In order to break a cyclic dependency, we move robots in some order. If robots move in some order, then the cyclic dependency breaks down, and further movements of the robots will break the remaining dependencies in that cycle (this part is achieved during Phase III). 
We obtain a partial ordering of the movements of the non-faulty
non-terminal robots using their radial distances $\rho_i(t)$ from
the reference point $P^*$, and move the farthest eligible
non-terminal robot first among those that have not yet completed
their Phase II movement.

We obtain a partial ordering of movements of the non-faulty non-terminal robots (faulty robots may not see all faulty non-terminal robots) by using their relative distances from the reference point, and we move the farthest non-terminal robot first among the non-terminal robots that have not moved yet in Phase II. Robots use their lights to communicate whether they have moved in Phase II or not (see below).

The Phase II ordering does not require all non-faulty non-terminal
robots to observe the same global set of candidates. It is sufficient
to compare robots that are directly dependent. In particular, two
consecutive critical non-terminal robots of a cyclic dependency are
friends on a common line of collinearity and hence are mutually
visible. Therefore, if two such robots $r_i$ and $r_j$ have
$\|r_i-P^*\|>\|r_j-P^*\|,$
then $r_j$ cannot be eligible while $r_i$ still has light
\textsc{moved-NT-1}, since $r_j$ observes the farther unprocessed
robot $r_i$. Hence, simultaneously eligible dependent robots must
have the same distance from $P^*$.
Now consider the case when all the non-terminal robots in a cyclic dependency are at the same distance $l$ from $P^*$. Then these points lie on a circle of radius $l$ centered at $P^*$(see Fig.~\ref{fig:ordering}(a)). Now, according to our distance rule, all of these robots are eligible to move in the same round. If they all move in the same round, their cyclic dependency may remain intact even after their movements. We move them along the line segment joining their respective current position and $P^*$ (this helps us in Phase III). However, due to these movements, one of the following will happen: either a non-terminal robot becomes terminal (when it and its friend move the same distance)(see Fig.~\ref{fig:ordering}(b)) or they remain cyclic dependent without all
lying on the same new circle having $P^*$ as center. In the second case, it can be easily shown that all these robots are at distinct distances from $P^*$(see Fig.~\ref{fig:ordering}(c)), and in Phase III, these robots move one by one according to their distances from $P^*$, and non-faulty non-terminal robots become terminal robots at the end of Phase III. 

During the Phase II, a non-terminal robot uses three colors for their lights: (i) {\it moving-2} to indicate that it has completed its movement in Phase II, (ii) if after the movement in Phase II, it remains non-terminal, it uses {\it moved-NT-2} to indicate that after the movement in Phase II, they are still non-terminals, and (iii)  if after the movement in Phase II, it becomes a terminal robot, it uses {\it moved-T-2} to indicate that after the movement in Phase II, it has become terminal. In Phase III, only the robots with light {\it moved-NT-2} move. During Phase II, each eligible non-faulty non-terminal robot $r_i$
moves towards the destination $\operatorname{dest}_i(t)$ lying in the
interior of the line segment joining $r_i(t)$ and $P^*$, as defined in
Section~\ref{comnputation}. We formally define the movement rules  for the non-terminal robots as follows: 
\noindent
\textbf{$r_i$ is a robot with light \textsc{moved-NT-1}:}
\begin{enumerate}
    \item If $r_i$ is currently terminal, it changes its light to
    \textsc{moved-T-2} without moving.

    \item If $r_i$ is currently non-terminal and is farthest from $P^*$
    among all the visible robots with light \textsc{moved-NT-1},
    then it changes its light to \textsc{moving-2} and moves toward
    its computed destination point on the line segment joining
    $r_i$ and $P^*$.

    \item If $r_i$ has light \textsc{moving-2}, then, upon its next
    activation, it changes its light to \textsc{moved-T-2} if it is
    currently terminal, and to \textsc{moved-NT-2} otherwise.
\end{enumerate}
In all other cases, $r_i$ does nothing.


After completing its Phase II action, a robot has light either
\textsc{moved-T-2} or \textsc{moved-NT-2}.
\begin{lemma}
\label{l_10}
After all Phase-II actions associated with a cyclic dependency have
completed, either the dependency is broken or every pair of
consecutive non-faulty critical non-terminal robots remaining in the
dependency have distinct distances from $P^*$.
\end{lemma}

\begin{proof}
Consider two consecutive non-faulty critical robots that remain
dependent after completing their Phase-II actions. They are friend
robots on their common line of collinearity. When the later of the
two completes its Phase-II movement and remains non-terminal,
Lemma~\ref{l_16} implies that its distance from $P^*$ differs from that
of each of its friends, in particular from the other critical robot.
Hence every consecutive remaining critical pair has distinct
distances from $P^*$. If no such pair remains, the dependency has
been broken.
\end{proof}




\begin{figure}[ht]
\centering
\includegraphics[width=0.5\linewidth]{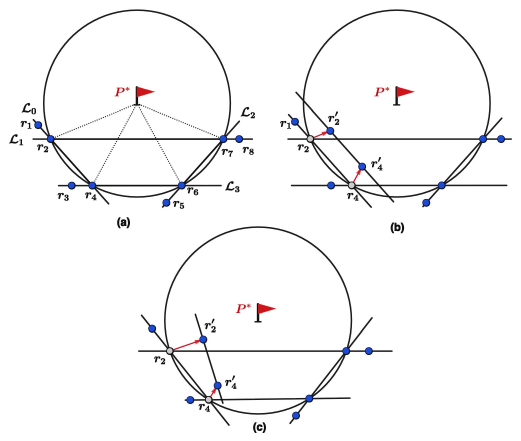}
\caption{An illustration of resolving an equal-distance cyclic dependency. (a) The critical non-terminal robots lie on a common circle centered at $P^*$. (b) If the robots move equal distances toward $P^*$, the existing collinearities are broken. (c) If they move unequal distances due to non-rigid movement, the collinearities are also broken, and the robots attain distinct distances from $P^*$, inducing an order for Phase III.}
\label{fig:ordering}
\end{figure}

\textbf{Phase III (Only non-faulty non-terminal robots move):}  
After the relevant Phase-II actions have completed, consecutive
non-faulty critical robots that remain in the same dependency have
distinct distances from $P^*$. Hence, the radial-distance ordering
prevents directly dependent robots from remaining simultaneously
eligible in Phase III.
Since after Phase II, the  non-faulty non-terminal robots belonging to the same cyclic or non-cyclic dependencies are at different distances from $P^*$, we can move them one by one. Since from each line of collinearity exactly one robot (among the three) moves, the remaining collinearities involving non-faulty non-terminal robots gradually break one by one. However, some collinearities may remain from Phase I  involving a faulty non-terminal robot and at least one non-faulty terminal robot (for example when an initial line of collinearity $\mathcal L'$ passes through $P^*$ two robots from $\mathcal L'$ may move in the two opposite sides of $\mathcal L'$ and they may become collinear on a new line with another faulty robot lying on $\mathcal L'$ in between). In this case, we handle the case in Phase IV.

During the Phase III, a non-terminal robot uses three colors for their lights: (i) {\it moving-3} to indicate that it has completed its movement in Phase III, (ii) if after the movement in Phase III, it remains non-terminal, it uses {\it moved-NT-3} to indicate that after the movement in Phase III, it is still a non-terminal robot (in case of faulty robots), and (iii)  if after the movement in Phase III, it becomes a terminal robot, it uses {\it moved-T-3} to indicate that after the movement in Phase III, it has become a terminal robot.  During Phase III, each eligible non-faulty non-terminal robot $r_i$
moves towards the destination $\operatorname{dest}_i(t)$ lying in the
interior of the line segment joining $r_i(t)$ and $P^*$, as defined in
Section~\ref{comnputation}. A robot with light \textsc{moved-NT-2} becomes eligible for its
Phase III action only when none of its friend robots on its current
dependency lines has light \textsc{moved-NT-1} or
\textsc{moving-2}. Thus, a Phase III action cannot overtake an
unfinished Phase II action on the same dependency, and no global
Phase II completion test is required.

\medskip

 \noindent {\bf $\boldsymbol{r_i}$ is a robot with light {\it moved-NT-2}:}
	\begin{enumerate}
		\item If $r_i$ is the farthest from the point $P^*$ among
		all the visible non-terminal robots with light {\it moved-NT-2}, then it changes its light to {\it moving-3}	and moves towards the computed destination point.

		\item If $r_i$ is a terminal robot and its light is {\it moving-3}, then $r_i$ changes its light to {\it moved-T-3} and does not move.

		\item  If $r_i$ is a non-terminal robot and its light is {\it moving-3}, then $r_i$ changes its light to {\it moved-NT-3} and does not move. In this case, robot $r_i$ is a faulty non-terminal robot.
	\end{enumerate}
	In all remaining cases, robot $r_i$ does nothing.

At the end of Phase III, every non-faulty robot has a light either {\it moved-T-1}, {\it moved-T-2}, {\it moved-T-3}, or
{\it moved-NT-3}. Moreover, all cyclic dependencies among the robots with light {\it moved-NT-2} are eliminated, and the non-faulty robots can identify the completion of this phase.\\


\textbf{Phase IV (Only terminal robots with light {moved-T-1}, {moved-T-2}, or {moved-T-3} move):}  Phase IV is necessary because a terminal non-faulty robot may still
have its visibility obstructed by a mobility-faulty non-terminal
robot. Such a configuration may arise after an earlier simultaneous
movement, while the faulty robot remains immobile. Therefore,
terminal robots with light \textit{moved-T-1},
\textit{moved-T-2}, or \textit{moved-T-3} perform one final radial
movement toward $P^*$. 
During Phase IV, each eligible non-faulty terminal robot $r_i$
moves towards the destination $\operatorname{dest}_i(t)$ lying in the
interior of the line segment joining $r_i(t)$ and $P^*$, as defined in
Section~\ref{comnputation}.

\noindent
$r_i$ is a physically terminal robot with light
\textsc{moved-T-1}, \textsc{moved-T-2}, or \textsc{moved-T-3}:

\begin{enumerate}
    \item Upon its Phase-IV activation, $r_i$ changes its light to
    \textsc{moving-4} and moves toward its radial destination
    $\operatorname{dest}_i(t)$.

    \item If $r_i$ has light \textsc{moving-4}, then upon its next
    activation it changes its light to \textsc{moved-T-4} and does
    not move.

    \item If $r_i$ has light \textsc{moved-T-4}, it performs no  further action.
\end{enumerate}

\textbf{Termination:}
If every robot visible to $r_i$, including $r_i$ itself, has light
either \textsc{moved-T-4} or \textsc{moved-NT-3}, then $r_i$
terminates without changing its current light color.

\subsection{Correctness}

Due to space constraints, we mainly state the lemmas that are essential to prove the correctness of our algorithm, \textsc{Mutual-Visibility-Mobility-Fault()}.

\begin{lemma}
\label{l_11}
During the whole execution of
\textsc{Mutual-Visibility-Mobility-Fault()}, each robot moves at most
four times.
\end{lemma}

\begin{proof}
A robot performs at most one movement in each of Phases I, II, III,
and IV. The light-state transitions prevent it from becoming eligible
for a second movement within the same phase. Therefore, every robot
performs at most four movements during the entire execution.
\end{proof}

 \begin{lemma}
 \label{l_12}
Let $r_i,r_j,r_k$ be three robots that are non-collinear in the
initial configuration. Then they remain non-collinear throughout the
execution of algorithm \textsc{Mutual-Visibility-Mobility-Fault()}.
\end{lemma}

\begin{proof}
Since $r_i,r_j,r_k$ are non-collinear before the considered round,
the current configuration is non-linear; hence the non-linear case
of the displacement rule applies.
We prove the claim by induction over the $\mathsf{SSYNC}$ rounds. Suppose that
$r_i,r_j,r_k$ are non-collinear immediately before a round, and let
$h>0$ be the minimum altitude of the triangle
$\triangle r_i r_j r_k$ at that time. Since these three robots are
non-collinear, the current robot configuration is non-linear.

We first show that every robot among $r_i,r_j,r_k$ that moves in the
round is displaced by less than $h/2$. Consider $r_i$. If both
$r_j$ and $r_k$ are visible to $r_i$, then the corresponding triangle
altitudes are considered in the computation of $d_i(t)$, and hence
$d_i(t)\leq h$. Therefore, by the displacement rule,
\[
        \sigma_i(t)\leq \frac{d_i(t)}{3^4}<\frac{h}{2}.
\]

Otherwise, let $\widehat r_j$ and $\widehat r_k$ be the first visible
robots from $r_i$ on the segments $r_ir_j$ and $r_ir_k$, respectively.
If $\widehat h$ denotes the minimum altitude of
$\triangle r_i\widehat r_j\widehat r_k$, then, as in the proof of
Lemma~\ref{l_5}, $ \widehat h\leq 2h.$ Since the altitudes of
$\triangle r_i\widehat r_j\widehat r_k$ are considered in computing
$d_i(t)$, $ d_i(t)\leq \widehat h\leq 2h,$
and consequently
\[
        \sigma_i(t)
        \leq \frac{d_i(t)}{3^4}
        \leq \frac{2h}{3^4}
        <\frac{h}{2}.
\]
The same argument applies to $r_j$ and $r_k$. A robot that is not
activated or is mobility-faulty has displacement zero, while a
non-rigid interruption can only reduce the executed displacement.
Thus, irrespective of which subset of the three robots moves
simultaneously in the round, every robot is displaced by less than
$h/2$.

Suppose, for contradiction, that after the simultaneous movements the
three robots become collinear on a line $\ell$. Since each new
position lies on $\ell$ and each displacement is less than $h/2$,
each pre-round position lies at perpendicular distance less than
$h/2$ from $\ell$. Hence the triangle
$\triangle r_i r_j r_k$ before the round is contained in an open
strip of width strictly smaller than $h$ whose axis is $\ell$.
This is impossible because the minimum width of a triangle equals
its minimum altitude, which is $h$.
Therefore, $r_i,r_j,r_k$ remain non-collinear after the round.
Induction over all $\mathsf{SSYNC}$ rounds proves the claim throughout the
execution.
\end{proof}



  \begin{lemma}
    \label{l_13}
Let $r_i\in\mathcal{R}$. If, during the execution of the algorithm, the movement of $r_i$ creates a new line of collinearity $\mathcal{L}$, then all robots lying on $\mathcal{L}$ originally belonged to a single line of collinearity from the initial configuration. Moreover, $\mathcal{L}$ contains exactly three robot positions.
\end{lemma}

\begin{proof}
Suppose that after the movement of $r_i$, the new line $\mathcal L$ contains
$r_i$ and two robots $r_j$ and $r_k$. If $r_i,r_j,r_k$ were not collinear
before the movement, Lemma~\ref{l_13} would prevent them from becoming collinear.
Therefore, every robot participating in the new collinearity must originate
from a line of collinearity through $r_i$.

Assume that robots originating from two distinct lines in
$\mathcal{L}_i^0$ lie on the same newly created line. Then at least
one triple consisting of $r_i$ and robots originating from these two
distinct initial lines was initially non-collinear, contradicting
Lemma~\ref{l_13}. Hence, all robots on the new line originate from a single
line $L^*\in\mathcal{L}_i^0.$

During Phase~I, at most two robots of an initial line move in one round, and
the movement directions are chosen so that a moving robot does not move along
an existing robot-to-robot line. The distance bound also prevents the moving
trajectory from intersecting two different lines determined by stationary
robots. Thus a new line can contain only the moving robot and at most two
robots originating from $\mathcal L^*$. Consequently, every newly created
line of collinearity contains exactly three robot positions.
\end{proof}

\begin{lemma}
 \label{l_15}
During Phase II, only non-faulty non-terminal robots are eligible to move. Moreover, every non-faulty robot that remains non-terminal eventually becomes eligible to move. \end{lemma}

\begin{proof}
The Phase~II movement rule applies only to a physically non-terminal robot
whose light is \textsc{moved-NT-1}. Hence terminal robots are not eligible.
A mobility-faulty robot may satisfy the logical eligibility condition and
change its light, but it cannot execute a displacement. Therefore, every
robot that actually moves in Phase~II is non-faulty and non-terminal.

Let $r_i$ be a non-faulty robot that remains non-terminal. All non-faulty
non-terminal robots are mutually visible at the end of Phase~I, by Lemma~\ref{l_9},
and therefore observe the same set of Phase~II candidates and their distances
from $P^*$. At least one candidate has maximum distance from $P^*$ and is
eligible. After every such candidate completes its Phase~II action, its light
is no longer \textsc{moved-NT-1}. Thus the number of unprocessed candidates
strictly decreases. Fairness guarantees that an eligible non-faulty candidate
is eventually activated. Repeating this argument shows that $r_i$ eventually
becomes eligible.
\end{proof}

\begin{lemma} 
 \label{l_16}
Let $r_i$ be a non-faulty robot that performs a Phase II movement. After completing the movement, either $r_i$ becomes terminal or, if it remains non-terminal, its distance from $P^*$ is different from the distance of each of its friend robots from $P^*$. 
\end{lemma}

\begin{proof}
Robot $r_i$ moves along the segment joining its current position to $P^*$.
Thus its distance from $P^*$ strictly decreases by the distance actually
traversed. If only $r_i$ moves on a line of collinearity, the old
collinearity is broken and $r_i$ becomes terminal with respect to that line.

The only relevant case in which $r_i$ can remain non-terminal is when one or
more dependent friend robots are activated in the same round. If the
movements produce equal new radial distances and preserve the local
dependency, the involved segments toward $P^*$ have equal executed lengths;
the geometry of the three-robot line then makes one of the participating
robots terminal. Otherwise, non-rigid interruptions yield unequal executed
lengths, and therefore unequal final distances from $P^*$. Hence, whenever
$r_i$ remains non-terminal, its distance from $P^*$ is different from that of
each friend involved in its remaining collinearities.
\end{proof}

\begin{lemma} 
 \label{l_17}
If non-faulty robots form a cyclic-dependency relation during Phase II and are simultaneously eligible to move, then all such robots are equidistant from $P^*$. Consequently, they lie on a common circle centered at $P^*$. 
\end{lemma}
\begin{proof}
Let $r_i$ and $r_j$ be two consecutive critical non-terminal robots
of the cyclic dependency. Since they are directly dependent on a
common line of collinearity, they are mutually visible.

Suppose that both are simultaneously eligible in Phase II. If
$\|r_i-P^*\|>\|r_j-P^*\|,$
then $r_j$ observes the farther unprocessed robot $r_i$ with light
\textsc{moved-NT-1}, and hence $r_j$ cannot satisfy the Phase II
farthest condition, a contradiction. Similarly,
$\|r_j-P^*\|>\|r_i-P^*\|$ is impossible. Therefore,
$
\|r_i-P^*\|=\|r_j-P^*\|.
$

Applying the same argument to every pair of consecutive critical
robots around the cyclic dependency yields $
\|r_1-P^*\|=\|r_2-P^*\|=\cdots=\|r_k-P^*\|.
$
Hence all simultaneously eligible critical robots lie on a common
circle centered at $P^*$.
\end{proof}



 \begin{lemma}
  \label{l_18}
Every non-faulty robot with light \textsc{moved-NT-1} eventually changes its
light to either \textsc{moved-T-2} or \textsc{moved-NT-2}.
\end{lemma}

\begin{proof}
Let $r_i$ be a non-faulty robot with light \textsc{moved-NT-1}. If $r_i$
becomes terminal before its Phase II movement, then at its next activation it
changes directly to \textsc{moved-T-2}. Otherwise, $r_i$ remains
non-terminal. On every current dependency line, its friend robots are visible.
Once these friends have completed their Phase I actions, the Phase II
eligibility rule applies.

Among any finite set of visible unprocessed candidates, at least one has
maximum distance from $P^*$. By fairness, an eligible non-faulty robot is
eventually activated and completes its Phase II action. After completing the
action it no longer has light \textsc{moved-NT-1}. Repeating the argument
eliminates all non-faulty robots with light \textsc{moved-NT-1} in finite
time.
\end{proof}




\begin{lemma}
 \label{l_19}
Every cyclic-dependency relation among non-faulty robots is eventually eliminated during Phase III.
\end{lemma}

\begin{proof}
At the end of Phase~II, either a cyclic dependency has already been broken or
its remaining non-faulty critical robots have pairwise distinct distances from
$P^*$, by Lemma~\ref{l_18}. Phase~III activates the farthest remaining robot
with light \textsc{moved-NT-2}. Therefore, among the non-faulty critical
robots of a dependency, at least one robot is selected before its dependent
neighbors.

The selected robot moves toward $P^*$. On each incident line of the cyclic
dependency, it is the only selected critical robot at that stage. Its
movement consequently breaks at least one dependency edge. By Lemma~\ref{l_13}, this
movement does not introduce a collinearity involving robots that were
previously non-collinear. After completing the movement, the robot changes
its light and is not selected again in Phase~III. Hence every effective
movement strictly decreases the number of unresolved dependency edges.
Because the number of such edges is finite, every cyclic dependency is
eventually eliminated.
\end{proof}
 
 \begin{lemma}
  \label{l_20}
Every non-faulty robot that is non-terminal at the beginning of
Phase III eventually becomes terminal.
\end{lemma}

\begin{proof}
Let $r_i$ be a non-faulty non-terminal robot at the beginning of Phase~III.
It has light \textsc{moved-NT-2}. By the distance ordering, the farthest
unprocessed non-terminal robot is eligible in every non-terminal
configuration. After its movement, that robot either immediately becomes
terminal or breaks at least one incident dependency.
Lemma~\ref{l_19} guarantees that cyclic dependencies cannot persist. Once the cyclic
parts are eliminated, the remaining dependencies form a finite non-cyclic
chains. Such a chain has an endpoint non-terminal robot whose movement breaks
its incident collinearity and makes it terminal. Removing this endpoint
shortens the chain. Induction on the number of robots in the dependency chain
shows that every non-faulty non-terminal robot eventually becomes terminal.
Fairness ensures that each eligible robot is activated after finitely many
rounds.
\end{proof}

\begin{lemma}
\label{l_21}
Every non-faulty robot with light \textsc{moved-NT-2} eventually performs its
Phase III action and leaves the state \textsc{moved-NT-2}.
\end{lemma}

\begin{proof}
A robot with light \textsc{moved-NT-2} waits only for unfinished Phase II
actions on its current dependency lines. These dependencies involve visible
friend robots. By Lemma~\ref{l_18}, every non-faulty such friend eventually
leaves its Phase II state; a mobility-faulty robot can still update its light
and therefore also eventually records completion of its Phase II action.
Hence the local waiting condition eventually disappears. The distance ordering
then makes at least one remaining robot eligible for Phase III, and fairness
guarantees its activation. Thus every non-faulty robot eventually leaves
\textsc{moved-NT-2}.
\end{proof}




\begin{lemma}
\label{l_22}
During Phase IV, every non-faulty terminal robot can see every robot in the system.
\end{lemma}
 \begin{proof}
At the beginning of Phase~IV, Lemma~\ref{l_20} and Lemma~\ref{l_21} implies that every non-faulty robot is
terminal. Thus any remaining visibility obstruction for a non-faulty robot
must be caused by a faulty robot lying between two terminal robots.
Let $r_i$ be an eligible non-faulty terminal robot that does not see any other robots. It moves toward a destination on the ray from $r_i$ toward $P^*$.
The Phase~IV construction guarantees that the current line of collinearity
through $r_i$ does not contain $P^*$. Hence, the radial movement of $r_i$
leaves that line immediately and remove the obstruction. By Lemma~\ref{l_13} and the
definition of $d_i(t)$, this movement does not create a collinearity with
robots that were previously non-collinear.
After the movement, $r_i$ changes its light to \textsc{moved-T-4} and does
not move again. Since every remaining obstructed non-faulty terminal robot
eventually becomes eligible and is activated by fairness, all such
obstructions are removed in finite time. Therefore, every non-faulty terminal
robot eventually sees every robot, including the mobility-faulty robots.
\end{proof}

\begin{lemma}
\label{l_23}
The movements of the robots are collision-free.
\end{lemma}

\begin{proof}
Consider an arbitrary $\mathsf{SSYNC}$ round starting at time $t$, and let
$\lambda_i(t)$ denote the distance actually traversed by robot $r_i$
during this round; set $\lambda_i(t)=0$ if $r_i$ does not move.
By the destination construction, $\lambda_i(t)\leq \sigma_i(t)
\leq \frac{D_i(t)}{3^4}.$ 

Consider any two distinct robots $r_i$ and $r_j$, and let $d_{ij}(t)=\|r_i(t)-r_j(t)\|.$ If $r_i$ sees $r_j$, then, by the definition of $D_i(t)$, $D_i(t)\leq d_{ij}(t).$ If $r_i$ does not see $r_j$, opacity implies that some robot lies in
the open segment $\overline{r_i(t)r_j(t)}$. Let $r_k$ be the first
such robot visible from $r_i$. Then $D_i(t)\leq \|r_i(t)-r_k(t)\|<d_{ij}(t).$ Thus, in either case, $\lambda_i(t)\leq \frac{d_{ij}(t)}{3^4}.$ The same argument gives $\lambda_j(t)\leq \frac{d_{ij}(t)}{3^4}.$ Consequently, $\lambda_i(t)+\lambda_j(t)
\leq \frac{2d_{ij}(t)}{3^4}
<d_{ij}(t).$

Suppose, for contradiction, that the two executed movement segments
intersect at a point $x$. Since $x$ lies on both executed segments, $d_{ij}(t)
\leq \|r_i(t)-x\|+\|x-r_j(t)\|
\leq \lambda_i(t)+\lambda_j(t)
<d_{ij}(t),$ a contradiction. This argument also covers the case in which one
robot is stationary, by taking its displacement to be zero.
Hence, the executed trajectories of every pair of robots are
disjoint in every $\mathsf{SSYNC}$ round. Non-rigid interruption only shortens
an executed segment and therefore cannot create a collision.
Thus, all movements are collision-free.
\end{proof}




  \begin{lemma}
  \label{l_24}
  The proposed algorithm Mutual-Visibility-Mobility-Fault() solves the problem in $O(n)$ epochs.
  \end{lemma}

\begin{proof}
In Phase I, every non-faulty robot performs exactly one Phase-I
movement, and eligibility propagates from the terminal robots toward
the interior of each dependency. Hence, by fairness, at least one
previously unprocessed robot completes its Phase-I action within a
constant number of epochs. Since there are at most $n$ robots,
Phase I requires $O(n)$ epochs.
In Phase II, every completed action removes one robot from the set
with light \textsc{moved-NT-1}; similarly, in Phase III, every
completed action removes one robot from the set with light
\textsc{moved-NT-2}. Since each such set contains at most $n$
robots and an eligible robot completes its action within a constant
number of epochs, each of Phases II and III requires $O(n)$ epochs.
Finally, every eligible terminal robot performs at most one
Phase-IV movement, so Phase IV also requires $O(n)$ epochs.
Therefore, the four phases together complete in $O(n)$ epochs.
\end{proof}

 \begin{lemma}
 \label{l_26}
 During the execution of Mutual-Visibility-Mobility-Fault, the total Euclidean distance traversed by every robot is at most its prescribed movement budget $B$. 
 \end{lemma}
\begin{proof}
For every movement of $r_i$, $\sigma_i(t)\leq \frac{B_i(t)}{3^4}<B_i(t),$ and only the actually traversed distance is deducted from the
remaining budget. If $\ell_i^{(q)}$ is the distance traversed in
the $q$-th movement, then $B_i^{(q+1)}=B_i^{(q)}-\ell_i^{(q)}\geq 0.$ Therefore, by telescoping, $\sum_q\ell_i^{(q)}
   =B-B_i^{(\mathrm{final})}\leq B.$ Hence, every robot respects its prescribed movement budget.
\end{proof}

\begin{theorem}
\label{t_1}
There exists a collision-free deterministic algorithm that solves the mutual visibility problem in $O(n)$ epochs for $\mathsf{SSYNC}$ robots with non-rigid, budget-constrained movements under the \textsc{FALL} light model using only $12$ light colors and tolerating an arbitrary number of mobility faults.
\end{theorem}

\begin{proof}
By Lemmas~\ref{l_18}--\ref{l_22}, every non-faulty robot eventually attains
unobstructed visibility of every robot in the system, including
the mobility-faulty robots. Lemma~\ref{l_23} guarantees collision-free
movements, Lemma~\ref{l_24} establishes the $O(n)$-epoch complexity,
and Lemma~\ref{l_26} shows that every robot respects its prescribed
movement budget. The algorithm uses only twelve light colors.
Hence \textsc{Mutual-Visibility-Mobility-Fault()} satisfies all
requirements of the budget-constrained strong fault-tolerant
mutual visibility problem.
\end{proof}

 \section{Conclusion}
\label{sec:conclusion}

We studied the mutual visibility problem under the budget-constrained mobility fault model for opaque luminous robots operating in the $\mathsf{SSYNC}$ model with non-rigid movements. Unlike existing algorithms, our solution neither relies on convex-hull decomposition nor assumes one-axis agreement. Instead, we showed that access to a common fixed external reference beacon is sufficient to coordinate the robots without requiring
axis agreement. We presented the first deterministic non-convex-hull-based algorithm for this model that guarantees collision-free mutual visibility from every admissible initial robot configuration satisfying the
reference-beacon assumption while respecting the prescribed movement budget and tolerating an arbitrary number of mobility faults. The algorithm uses only 12 light colors and terminates in $O(n)$ epochs. Future work includes extending the proposed algorithm to the $\mathsf{ASYNC}$ model and investigating whether the mutual visibility problem can be solved under the budget-constrained mobility fault model without assuming a common fixed reference point and any form of local axes agreement.

\bibliographystyle{plain} 
\bibliography{Biblo}

\end{document}